\documentclass{snapshotmfo}
\usepackage[utf8]{inputenc}
\usepackage[dvipsnames]{xcolor}
\usepackage{mathtools}
\usepackage{enumitem}
\usepackage{amsmath,amssymb}
\usepackage{amsthm}
\usepackage{tcolorbox}
\usepackage{epigraph}
\usepackage{hyperref}
\hypersetup{
    colorlinks=true,
    linkcolor=Aquamarine,
    filecolor=Aquamarine,      
    urlcolor=Aquamarine,
    citecolor=Aquamarine,
}

\usepackage[round]{natbib}
\usepackage[USenglish]{babel}

\usepackage{ellipsis}

\DeclareMathOperator\supp{supp}
\DeclareMathOperator\inter{int}

\newtheorem{theorem}{Theorem}[section]
\newtheorem{proposition}[theorem]{Proposition}

\newtheorem{remark}[theorem]{Remark}
\newtheorem{corollary}[theorem]{Corollary}

\newtheorem{definition}[theorem]{Definition}

\author{Mark Whitmeyer \and Cole Williams\thanks{MW: Arizona State University, \href{mailto:mark.whitmeyer@gmail.com}{mark.whitmeyer@gmail.com}. CW: Durham University \href{mailto:cole.randall.williams@gmail.com}{cole.randall.williams@gmail.com}. Draft date: \textcolor{OrangeRed}{\today}.}}
\title{Comparisons of Sequential Experiments for Additively Separable Problems}

\begin{document}

\begin{abstract}
For three natural classes of dynamic decision problems; 1. additively separable problems, 2. discounted problems, and 
3. discounted problems for a fixed discount factor; we provide necessary and sufficient conditions for one sequential experiment to dominate another in the sense that the dominant experiment is preferred to the other for any decision problem in the specified class. We use these results to study the timing of information arrival in additively separable problems.
\end{abstract}

\section{Introduction}

Many decision problems of interest are dynamic: there is some (possibly infinite) time horizon and as time passes by the the agent obtains information through various sources and takes multiple actions at different instances. A natural question is how to rank information for the agent in such a dynamic setting. How should one rank dynamic information structures? \cite*{greenshtein1996comparison} answers this with aplomb (see also \cite*{de2018blackwell}): echoing \cite*{blackwell}'s seminal ranking, he provides a partial order over dynamic information structures in which one dominates another if and only an agent prefers the one to the other for any dynamic decision problem.

However, the class of dynamic decision problems is large and variegated and this is reflected by \citeauthor{greenshtein1996comparison}'s order. It has little bite and is also difficult to parse--for instance, it is difficult to summarize what precisely dominance in his order means in terms of the agent's sequence of beliefs. To sidestep these issues, in this paper we restrict attention to several important subclasses of dynamic decision problems. In each, we provide necessary and sufficient conditions for an agent to prefer one information structure to another, for any decision problem within the specified class. The conditions are simple, intuitive, and economically meaningful.

The first, and broadest subclass we consider are the (intertemporally) additively separable decision problems. These are those in which the agent's payoff in the \(T\)-period decision problem equals \(\sum_{t=1}^T u_t(a_t,\theta)\) for some collection of per-period utility functions \(u_t\). We find that within this class, one dynamic information structure, \(\pi_1\), is preferred to another, \(\pi_2\), if and only if at each period \(t\), the agent's period-\(t\) distribution over posteriors induced by \(\pi_1\) is a mean-preserving spread of that induced by \(\pi_2\).

We then specialize further, to (exponentially-)discounted decision problems for a fixed discount factor \(\delta\). These are those in which the agent's payoff in the \(T\)-period decision problem equals \(\sum_{t=1}^T \delta^{t-1}v(a_t,\theta)\) for some utility function \(v\). In this class, \(\pi_1\) is preferred to \(\pi_2\) if and only if the ``average'' distribution over posteriors induced by \(\pi_1\)--where the averaging is taken with respect to distribution induced by the weighted sums of the geometric series corresponding to \(\delta\)--is a mean-preserving spread of that induced by \(\pi_2\). 

An easy implication of this ranking is a comparison of dynamic experiments for discounted decision problems where the discount factor is not fixed. Quite simply, there, \(\pi_1\) is preferred to \(\pi_2\) if and only if the ``average'' distribution over posteriors induced by \(\pi_1\) is a mean-preserving spread of that induced by \(\pi_2\) for any ``averaging'' produced by some \(\delta\).

We finish by using these results to study an agent's preferences over the random arrival of information. Suppose there is a single static experiment, the realization of which arrives at a random time. Can we compare different arrival rates? Yes. In particular, in the class of additively separable problems, an agent prefers one arrival rate to another if and only if it is stochastically earlier: the distribution over arrival times preferred by an agent is first-order stochastically dominated by the less-preferred one. Notably, this means that in the broad class of additively separable problems, agents are neither risk-averse nor risk-loving with respect to the arrival time. On the other hand, we show that in the class of discounted decision problems, an agent is risk-loving with respect to the timing of the arrival of information: she prefers mean-preserving spreads of arrival times.

\section{Model}

An agent is faced with an \textcolor{PineGreen}{Additively-separable dynamic decision problem}, i.e., an \textcolor{PineGreen}{AS problem} or one in the \textcolor{PineGreen}{AS class}. There is a finite set of possible states of the world \(\Theta\). Time is discrete: there are \(T\) periods, where \(1 \leq T \leq \infty\), and each period is indexed by \(t \in \left\{1,\dots,T\right\} \eqqcolon \mathcal{T}\). Each period \(t\), the agent has access to a set of actions \(A_t\). We specify that for each \(t\), \(A_t\) is compact and denote \(A \coloneqq \times_{t=1}^T A_t\). The decision problem being in the AS class means that the agent's utility function \(u \colon A \times \Theta \to \mathbb{R}\) can be written
\[u(a,\theta) = \sum_{t=1}^T u_t(a_t,\theta)\text{,}\]
where for each \(t\) \(u_t\) is continuous and the period-\(t\) utilities are such that the sum is finite.

A notable subclass of AS problems are the \textcolor{PineGreen}{Discounted problems}. For a problem in the discounted class, each \(A_t\) equals some common set \(A^*\) and
\[u(a,\theta) = \sum_{t=1}^T \beta_t v(a_t,\theta)\text{,}\]
for some sequence \(\beta \coloneqq \left(\beta\right)_{t=1}^T\) for which \(\bar{\beta} \coloneqq \sum_{t=1}^T \beta_t\) is finite, and for some continuous utility function \(v \colon A^* \times \Theta \to \mathbb{R}\). A subclass of these problems are the \textcolor{PineGreen}{Exponentially discounted problems}, which are those for which \(\beta_t = \delta^{t-1}\) for each \(t\); for some \(\delta \in \left[0,1\right]\) if \(T \in \mathbb{N}\) or \(\delta \in \left[0,\bar{\delta}\right)\) (\(\delta \in \left[0,1\right)\)) if \(T = \infty\).

A subclass of the discounted problems are the \(\beta-\)\textcolor{PineGreen}{Discounted problems}, which are those for which 
\[u(a,\theta) = \sum_{t=1}^T \beta_t v(a_t,\theta)\text{,}\]
for a \textit{specified} sequence \(\beta\). A further subclass of these are the \(\delta-\)\textcolor{PineGreen}{Discounted problems}, where the exponential discount factor \(\delta\) is specified.

Each period the agent obtains information. Formally, there is a collection of compact sets of signal realizations \(\left\{S_1, \dots, S_T\right\}\). Then, letting \(S \coloneqq \times_{t=1}^T S_t\), a \textcolor{PineGreen}{Dynamic information structure} is a stochastic map \(\pi \colon \Theta \to \Delta (S)\).\footnote{Henceforth we say just \textcolor{PineGreen}{Information structure} or \textcolor{PineGreen}{Signal}.} Each period \(t\), the agent observes a realization \(s_t \in S_t\) according to \(\pi\) before choosing an action \(a_t \in A_t\).

We denote \(S^{t} \coloneqq S_1 \times \cdots \times S_t\) and define \(\pi_t \in \Delta\left(S^t\right)\) to be the marginal distribution over \(S^t\) induced by \(\pi\). A strategy for an agent is a collection of mappings \(\left(\alpha_t\right)_{t=1}^{T}\), where \(\alpha_t \colon S^t \to \Delta\left(A_t\right)\) for all \(t\). Note that in principle, an agent's strategy could also depend on her history of actions; however, the additive separability of the decision problem makes it without loss of optimality to forbid such dependence. The agent's (\textit{ex ante}) expected utility given prior \(\mu \in \inter \Delta\left(\Theta\right)\), information structure \(\pi\), and strategy \(\left(\alpha_t\right)_{t=1}^{T}\) is
\[\sum_{t=1}^T \sum_{\theta \in \Theta} \sum_{s^t \in S^t} \sum_{a_t \in A_t} \mu(\theta) \pi_t(s^t) \alpha_t (\left.a_t\right|s^t)u_t(a_t,\theta)\text{.}\]

For each \(t\), strategy \(\alpha_t\) is optimal if it maximizes
\[\sum_{\theta \in \Theta} \sum_{s^t \in S^t} \sum_{a_t \in A_t} \mu(\theta) \pi_t(s^t) \alpha_t (\left.a_t\right|s^t)u_t(a_t,\theta)\text{.}\]
By the compactness of each \(A_t\) and the continuity of each \(u_t\), a maximizer exists and we denote an arbitrary such maximizer \(\alpha^*_t\). Thus, the value of a signal, \(W(\pi)\), is
\[W(\pi) \coloneqq \sum_{t=1}^T \sum_{\theta \in \Theta} \sum_{s^t \in S^t} \sum_{a_t \in A_t} \mu(\theta) \pi_t(s^t) \alpha^{*}_t (\left.a_t\right|s^t)u_t(a_t,\theta)\text{.}\]

\begin{definition}
    Given a specified class of dynamic decision problem, \textcolor{PineGreen}{Signal \(\pi_1\) dominates signal \(\pi_2\)} if for any decision problem in the class \(W(\pi_1) \geq W(\pi_2)\). We write this \(\pi_1 \trianglerighteq \pi_2\).
\end{definition}

\subsection{Preliminaries}

Given full-support prior \(\mu \in \Delta\left(\Theta\right)\), any signal induces a sequence of of random vectors supported on a subset of \(\Delta\left(\Theta\right)\) with cumulative distribution functions (cdfs) \(F_1, \dots, F_T\), defined as follows. Formally, for each \(t\), \(F_t \in \Delta\left(\Delta\left(\Theta\right)\right)\) is the distribution over posteriors in period \(t\) induced by signal \(\pi^t\). Letting \(\succeq\) denote the static mean-preserving spread (MPS) relation (\(F \succeq G\) means \(F\) is an MPS of \(G\)), we have
\begin{remark}
    \(F_T \succeq F_{T-1} \succeq \dots \succeq F_1\).
\end{remark}
This is an immediate consequence of the fact that each period the agent is acquiring new information, refining her beliefs.

In each period \(t\), and for any belief \(x \in \Delta\left(\Theta\right)\) we define
\[V_t(x) \coloneqq \max_{a_t \in A_t} \mathbb{E}_x u_t(a_t,\theta)\text{.}\]
Consequently, we can rewrite the value of a signal in an AS problem as 
\[W(\pi) = \sum_{t=1}^T \mathbb{E}_{F_t} V_t(x)\text{.}\]
Moreover, in a \(\beta\)-discounted problem
\[W(\pi) = \sum_{t=1}^T \beta_t \mathbb{E}_{F_t} V(x)\text{,}\]
where \(V(x) \coloneqq \max_{a_t \in A^*} \mathbb{E}_x v(a_t,\theta)\).

\section{Results}

Our first result is that in the additively-separable class of dynamic decision problems, dominance is equivalent to dominance in the convex order of the distribution over posteriors induced by the signal, \textit{in each period}. Sequence \(\left(F_t\right)_{t=1}^T\) is that induced by \(\pi_1\) and \(\left(G_t\right)_{t=1}^T\) is that induced by \(\pi_2\). Then,
\begin{theorem}\label{theorem1}
    In the AS class, \(\pi_1 \trianglerighteq \pi_2\) if and only if \(F_t \succeq G_t\) for every \(t \in \mathcal{T}\).
\end{theorem}
\begin{proof}
    \(\left(\Leftarrow\right)\) If \(F_t \succeq G_t\) for each \(t\), \(\mathbb{E}_{F_t} V_t(x) \geq \mathbb{E}_{G_t} V_t(x)\) for each \(t\), which implies \(W(\pi_1) \geq W(\pi_2)\).

    \medskip

    \noindent \(\left(\Rightarrow\right)\) Now suppose for the sake of contraposition that there exists a \(t' \in \left\{1,\dots,T\right\}\) for which \(F_{t'} \not\succeq G_{t'}\). This implies that there exists a \(V_{t'}\) for which \(\mathbb{E}_{F_{t'}} V_{t'}(x) < \mathbb{E}_{G_{t'}} V_{t'}(x)\). Then, take a decision problem with \(V_t = 0\) for all \(t \neq t'\), so
    \[W(\pi_1) = \sum_{t=1}^T \mathbb{E}_{F_t} V_t(x) = \mathbb{E}_{F_{t'}} V_{t'}(x) < \mathbb{E}_{G_{t'}} V_{t'}(x) = \sum_{t=1}^T \mathbb{E}_{G_t} V_t(x) = W(\pi_2)\text{,}\]
    and so \(\pi_1\) does not dominate \(\pi_2\).\end{proof}
Now let us turn attention to \(\beta\)-discounted problems. A central object in ranking signals will be the following distribution. For \(t \in \mathcal{T}\), we construct the  probability mass function 
\[\lambda_\beta(t) = \frac{\beta_t}{\bar{\beta}}\text{,}\]
which specializes to the following in \(\delta\)-discounted problems:
\[\lambda(t) = \begin{cases}
\delta^{t-1}\frac{1-\delta}{1-\delta^{T}} \quad &\text{if} \quad \delta \in \left[0,1\right)\\
\frac{1}{T} \quad &\text{if} \quad \delta = 1 \ \text{and} \ T \in \mathbb{N}\text{.}
\end{cases}\]
Define cdfs \[F^{\beta} \coloneqq \sum_{t=1}^{T}\lambda_\beta(t)F_t \quad \text{and} \quad G^{\beta} \coloneqq \sum_{t=1}^{T}\lambda_\beta(t)G_t\text{,}\]
recalling that the \(F_t\)s are produced by \(\pi_1\) and the \(G_t\)s are produced by \(\pi_2\).
\begin{theorem}\label{discountedtheorem}
    In the \(\beta-\)Discounted class, \(\pi_1 \trianglerighteq \pi_2\) if and only if \(F^{\beta} \succeq G^{\beta}\).
\end{theorem}
\begin{proof}
    Fix arbitrary \(T\), \(\beta\), and \(\mu \in \Delta\left(\Theta\right)\), and observe that
    \[\begin{split}
        W(\pi_1) \geq W(\pi_2) &\Leftrightarrow \sum_{t=1}^T \beta_t\mathbb{E}_{F_t} V(x) \geq \sum_{t=1}^T \beta_t\mathbb{E}_{G_t} V(x)\\
        &\Leftrightarrow \frac{1}{\bar{\beta}}\sum_{t=1}^T \beta_t\mathbb{E}_{F_t} V(x) \geq \frac{1}{\bar{\beta}}\sum_{t=1}^T \beta_t\mathbb{E}_{G_t} V(x)\\
        &\Leftrightarrow \mathbb{E}_{\sum_{t=1}^T \frac{\beta_t}{\bar{\beta}} F_t} V(x) \geq \mathbb{E}_{\sum_{t=1}^T \frac{\beta_t}{\bar{\beta}} G_t} V(x)\\
        &\Leftrightarrow \mathbb{E}_{F^{\beta}} V(x) \geq \mathbb{E}_{G^{\beta}} V(x)\text{,}
    \end{split}\]
    where the third equivalence was by the linearity of expectation and the fourth by the definitions of \(F^{\beta}\) and \(G^{\beta}\). \end{proof}

\begin{corollary}
    In the Discounted class, \(\pi_1 \trianglerighteq \pi_2\) if and only if \(F^{\beta} \succeq G^{\beta}\) for any sequence \(\beta\) whose sum is finite.
\end{corollary}
\begin{proof} \(\left(\Rightarrow\right)\) Suppose for the sake of contraposition that there exists a finite-sum \(\beta\) for which \(F^{\beta} \not\succeq G^{\beta}\). Then Theorem \ref{discountedtheorem} implies that in the \(\beta\)-Discounted class \(\pi_1 \not\trianglerighteq \pi_2\).

\medskip

\noindent \(\left(\Leftarrow\right)\) Immediate by construction. \end{proof}

\section{Not What But When: Time Lotteries for Information}

Take some (static) non-trivial information structure \(\xi \colon \Theta \to \Delta(Z)\), where \(Z\) is a finite set of signal realizations; \(Z = \left\{z_1,\dots,z_m\right\}\). A special variety of dynamic information structure is one in which the realization of \(\xi\) arrives at a single random time and no information arrives otherwise. In this section, we apply our earlier results to derive two findings. First, in the AS class of dynamic decision problems, we show that (stochastically) earlier arrival of information is better for an agent. Second, in the discounted class of problems, we show that an agent is risk-loving over the arrival time of information.

Let \(h\) be the probability mass function of the information arrival random variable, \(Y\), supported on \(\left\{1,\dots,T\right\}\), where \(T \leq \infty\). That is \(h(y) = \mathbb{P}(Y = y)\) for all \(y \in \left\{1,\dots,T\right\}\). Let \(H\) denote the cdf of \(Y\): \[H(y) \coloneqq \mathbb{P}(Y \leq y) = \sum_{i=1}^{y}h(i)\text{.}\]

Cdf \(P\) first-order stochastically dominates (FOSD) cdf \(H\) if and only if \(H(y) \geq P(y)\) for all \(y \in \left\{1,\dots,T\right\}\). Cdf \(P\) second-order stochastically dominates (SOSD) cdf \(H\) if and only if the information-arrival random variable \(Y_1\) with cdf \(H\) is a mean-preserving spread of information-arrival random variable \(Y_2\) (with cdf \(P\)). That is, there exists a mean-zero random variable \(W\) such that \(Y_1 = Y_2 + W\).

For two (discrete) random variables \(Z\) and \(Y\) supported on a subset of \(\left\{1,\dots,T\right\}\), \(Z\) is a \textcolor{PineGreen}{Binary splitting} of \(Y\) if there exist \(z_1, z_3 \in \supp Z\), \(y_2 \in \supp Y\), and \(\eta_1, \eta_3 \in \mathbb{R}_{++}\) for which
\begin{enumerate}
    \item \(z_1 < y_2 < z_3\);
    \item \(\eta_1 z_1 +  \eta_3 z_3 = \left(\eta_1 + \eta_3\right)y_2\);
    \item  \(\mathbb{P}(Z = z) = \mathbb{P}(Y = z)\) for all \(z \in \supp Z \setminus \left\{z_1, z_3, y_2\right\}\);
    \item \(\mathbb{P}(Z = z_1) = \mathbb{P}(Y = z_1) + \eta_1\) and \(\mathbb{P}(Z = z_3) = \mathbb{P}(Y = z_3) + \eta_3\); and
    \item \(\mathbb{P}(Z = y_2) = \mathbb{P}(Y = y_2) - \eta_1 - \eta_3\).
\end{enumerate}

\begin{remark}\label{binary}
    Information arrival random variable \(Y_1\) is a mean-preserving spread of \(Y_2\) if and only if there exists a sequence of binary splittings \(Y_2, Z_1, \dots, Z_k, Y_1\).
\end{remark}
\begin{proof}
    This is implied by Theorem 2 in \cite*{ROTHSCHILD1970225} and Theorem 1a in \cite*{rasmusen1992defining}.
\end{proof}
We use this result to prove that a \(\beta-\)discounter is risk-loving over the random arrival of information, provided \(\beta\) is a decreasing sequence. Before proving that, we establish that earlier-arriving information is better for additively separable problems.
\begin{proposition}
    \(F_t \succeq G_t\) for every \(t\) if and only if \(P\) FOSD \(H\).
\end{proposition}
\begin{proof}
Let \(\rho \in \Delta\left(\Delta\left(\Theta\right)\right)\) be the distribution over posteriors produced by static information structure \(\xi\) and denote its support by \(X \coloneqq \left\{x^1, \dots, x^m\right\}\), where each \(x^i \in \Delta\left(\Theta\right)\).

Observe that, for each \(x \in X \cup \left\{\mu\right\}\).
    \[f_t(x) = \begin{cases}
        H(t) \rho(x) \quad &\text{if} \quad x \neq \mu\\
        H(t) \rho(x) + (1-H(t)) &\text{if} \quad x = \mu\text{.}
    \end{cases}\]
    Keep in mind that for each \(t\), \(f_t\) is a probability mass function (pmf): \(f_t(x) = \mathbb{P}(x)\). \(g_t\) and \(G_t\) are defined analogously.

    \noindent \(\left(\Rightarrow\right)\) Suppose for the sake of contraposition that \(P\) does not FOSD \(H\). Then there exists a \(t'\) such that \(P(t') > H(t')\). Then for all \(x \in X \setminus \left\{\mu\right\}\), 
    \[f_{t'}(x) = H(t') \rho(x) < P(t') \rho(x) = g_{t'}(x)\text{,}\]
    so \(F_{t'} \not\succeq G_{t'}\). 

\medskip

\noindent \(\left(\Leftarrow\right)\) Suppose \(P\) FOSD \(H\).Then for all \(t\) and for all \(x \in X \setminus \left\{\mu\right\}\), 
    \[f_{t}(x) = H(t) \rho(x) \geq P(t) \rho(x) = g_{t}(x)\text{,}\]
    so \(F_{t} \succeq G_{t}\) for all \(t\). \end{proof}
This proposition plus Theorem \ref{theorem1} imply
\begin{corollary}
    In the AS class, \textbf{Earlier is Better:} \(\pi_1 \trianglerighteq \pi_2\) if and only if \(P\) (corresponding to \(\pi_2\)) FOSD \(H\).
\end{corollary}
Now let us turn our attention to discounted problems. The sequence \(\beta = \left(\beta_t\right)_{t=1}^T\) is decreasing if \(\beta_1 \geq \dots \geq \beta_T\). Note that for exponential discounters, this sequence is decreasing.
\begin{proposition}\label{propsosd} Fix \(\mu \in \inter \Delta\left(\Theta\right)\); and let \(T \in \mathbb{N}\) (\(T \geq 3\)) and fix a decreasing sequence \(\beta\). If \(P\) SOSD \(H\), \(F^{\beta} \succeq G^{\beta}\).
\end{proposition}
\begin{proof}
By Remark \ref{binary}, we may assume without loss of generality that \(H\) is obtained from \(P\) by a binary splitting. We have
    \[f^{\beta}(x) = \begin{cases}
        \sum_{t=1}^{T}\lambda_{\beta}\left(t\right) H(t) \rho(x) \quad &\text{if} \quad x \neq \mu\\
        \sum_{t=1}^{T}\lambda_{\beta}\left(t\right) (H(t) \rho(x) + (1-H(t))) &\text{if} \quad x = \mu\text{.}
    \end{cases}\]
    This is also a pmf. \(g^{\beta}\) is constructed analogously from \(g_t\).
    
    It suffices to show that for all \(x \in X \setminus \left\{\mu\right\}\), \(g^{\beta}(x) = \tau_x f^{\beta}(x)\) for some collection \(\left(\tau_x\right)_{x \in X}\) where for each \(x \in X\), \(\tau_x \in \left[0,1\right]\). Observe that for each \(x \neq \mu\),
    \[\begin{split}
        f^{\beta}(x) - g^{\beta}(x) &= \sum_{t=1}^{T}\lambda_{\beta}\left(t\right) H(t) \rho(x) - \sum_{t=1}^{T}\lambda_{\beta}\left(t\right) P(t) \rho(x)\\
        &= \sum_{t=1}^{T}\lambda_{\beta}\left(t\right) (H(t)-P(t)) \rho(x)
    \end{split}\]
    Recall that \(H\) is obtained via an arbitrary binary splitting of \(P\). Let \(h\) and \(p\) denote their respective pmfs. We pick a \(t^* \in \left\{2,\dots,T-1\right\}\) and subtract \(k_{t^*} \in \left(0,p(t^*)\right]\) from \(p(t^*)\). We add \(k_a > 0\) to some \(a < t^*\) and \(k_b > 0\) to some \(b > t^*\), where \(t^* k_{t^*} = a k_a + b k_b\) and \(k_{t^*} = k_a + k_b\). Then, for all \(t < a\), \(H(t) = P(t)\); for \(a \leq t < t^*\), \(H(t) = P(t) + k_a\); for \(t^* \leq t < b\),  \(H(t) =  P(t) + k_a - k_{t^*}\) and for \(t^* \geq b\) \(H(t) = P(t)\).

    Consequently,
    \[\sum_{t=1}^{T}\lambda_{\beta}\left(t\right) (H(t)-P(t)) = \sum_{t=a}^{t^{*}-1}\lambda_{\beta}\left(t\right) k_a + \sum_{t=t^{*}}^{b-1}\lambda_{\beta}\left(t\right) (k_a-k_{t^*})\text{.}\]
    Furthermore, as \(k_{t^*} > k_a > 0\) and \(\lambda_\beta\) is decreasing in \(t\), \(\sum_{t=a}^{t^{*}-1} k_a + \sum_{t=t^{*}}^{b-1} (k_a-k_{t^*}) \geq 0\) implies \(\sum_{t=a}^{t^{*}-1}\lambda_{\beta}\left(t\right) k_a + \sum_{t=t^{*}}^{b-1}\lambda_{\beta}\left(t\right) (k_a-k_{t^*}) \geq 0\). Finally, 
    \[\begin{split}
        \sum_{t=a}^{t^{*}-1} k_a + \sum_{t=t^{*}}^{b-1} (k_a-k_{t^*}) &= (t^{*}-a)k_a + (b-t^{*})(k_a-k_{t^*}) =  b (k_b + k_a -k_{t^*}) = 0\text{.}
    \end{split}\]
    We conclude the result. \end{proof}
    This proposition plus Theorem \ref{discountedtheorem} imply
    \begin{corollary}
        In the \(\beta-\)Discounted class, if \(\beta\) is a decreasing sequence, \textbf{The DM is Risk-Loving Over Information Arrival:} if \(P\) SOSD \(H\), \(\pi_1 \trianglerighteq \pi_2\).
    \end{corollary}
    and
    \begin{corollary}
        In the \(\delta-\)Discounted class, \textbf{The DM is Risk-Loving Over Information Arrival.}
    \end{corollary}

    It has been observed that an exponential discounter must be risk-loving over the random arrival of a monetary prize (\cite*{chesson2003commonalities}, \cite*{chen2013effect}, and \cite*{dejarnette2020time}). Here, we show that an exponential discounter must be risk-loving over the random arrival of information, no matter what (non-trivial) information it is. Importantly, note that these two findings are not the same: unlike a prize realized from a lottery, whose benefit accrues only in the period it is obtained--or, at the very least, if the agent can save, can be spread over the periods following its arrival--once an agent has acquired information, it is helpful from then on.

    Moreover, Proposition \ref{propsosd} obtains irrespective of the shape of the decreasing sequence \(\beta\). It is tempting, upon observing the result for \(\delta\)-discounters for whom \(\lambda(t)\) is a convex function, to deduce that this risk-loving result is a consequence of this convexity. Not so: all that is needed is that the agent values the future less than the present. Indeed it is easy to construct an example in which the agent is strictly hurt by a binary splitting for some increasing sequence \(\beta\), even if it is convex (in \(t\)).

\bibliography{sample.bib}

\end{document}